\documentclass[conference]{IEEEtran}
\usepackage{amssymb,amsthm}
\usepackage{amsmath,amssymb}
\usepackage{bm,amsfonts}
\usepackage{nidanfloat}
\usepackage{enumerate}
\usepackage{algorithm,algorithmic}
\usepackage{graphicx}
\usepackage{comment}

\newcommand{\E}{\mathrm{E}}

\newcommand{\n}{\nonumber}
\newcommand{\nn}{\nonumber\\}

\newcommand{\rd}{\mathrm{d}}

\newcommand{\cU}{\mathcal{U}}
\newcommand{\cX}{\mathcal{X}}
\newcommand{\cI}{\mathcal{I}}
\newcommand{\cO}{\mathcal{O}}

\newcommand{\leg}{\prec}
\newcommand{\geg}{\succ}

\newcommand{\com}{\,,}
\newcommand{\per}{\,.}

\newtheorem{lemma}{Lemma}
\newtheorem{theorem}{Theorem}

\newenvironment{proof1}{\noindent{\it Proof.}}{\vspace{-5.5mm}}
\newenvironment{proof2}[1]{\noindent{\it Proof #1.}}{\vspace{-5.5mm}}

\def\dqed{\relax\tag*{\qed}}

\newcommand{\lprec}[1]{#1-1}
\newcommand{\gprec}[1]{#1-1}

\newcommand{\FG}{F}
\newcommand{\tF}{\tilde{F}}
\newcommand{\bF}{\bar{F}}

\newcommand{\tx}{\tilde{x}}
\newcommand{\als}{\gamma}
\newcommand{\alh}{\gap}
\newcommand{\lh}{l_I}
\newcommand{\ls}{\overline{l}}

\newcommand{\tp}{\tilde{p}}
\newcommand{\gap}{\rho}
\newcommand{\gray}{g}
\newcommand{\precg}{\prec}
\newcommand{\precl}{\prec_{\mathrm{L}}}
\newcommand{\hx}{\hat{x}}
\newcommand{\hl}{\hat{l}}
\newcommand{\hF}{\hat{F}}
\newcommand{\sfeg}{\mathrm{SFEG}}

\newcommand{\oF}{\overline{F}}
\newcommand{\oFI}{\overline{F_I}}
\newcommand{\res}[2]{\langle#1\rangle_{#2}}
\begin{document}
%
\title{Variable-to-Fixed Length Homophonic Coding with a Modified Shannon-Fano-Elias Code}


\author{\IEEEauthorblockN{Junya Honda \qquad Hirosuke Yamamoto}
\IEEEauthorblockA{Department of Complexity Science and Engineering\\
The University of Tokyo,\\
Kashiwa-shi Chiba 277--8561, Japan\\
Email: honda@it.k.u-tokyo.ac.jp, Hirosuke@ieee.org}
}

\maketitle

\begin{abstract}
Homophonic coding is a framework to reversibly convert a message
into a sequence with some target distribution.
This is a promising tool to generate a codeword with
a biased code-symbol distribution,
which is required for capacity-achieving communication by asymmetric channels.
It is known that
asymptotically optimal homophonic coding can be realized
by a Fixed-to-Variable (FV) length code using an interval algorithm
similar to a random number generator.
However, FV codes are not preferable
as a component of channel codes since
a decoding error propagates to all subsequent codewords.
As a solution for this problem
an asymptotically optimal Variable-to-Fixed (VF) length homophonic code, dual
Shannon-Fano-Elias-Gray (dual SFEG) code,
is proposed in this paper.
This code can be interpreted as a dual of
a modified Shannon-Fano-Elias (SFE) code based on Gray code.
It is also shown as a by-product that
the modified SFE code, named SFEG code, 
achieves a better coding rate
than the original SFE code in lossless source coding.
\end{abstract}


%
\IEEEpeerreviewmaketitle

\allowdisplaybreaks[4]

\section{Introduction}
In the communication through asymmetric channels,
it is necessary to use codewords
with a biased code-symbol distribution maximizing
the mutual information between the input and the output
to achieve the capacity.
It is well known that
biased codewords can be generated
from an auxiliary code over an extended alphabet
based on Gallager's nonlinear mapping \cite[p.\,208]{gallager_map},
but its complexity becomes very large
when the target distribution is not expressed
in a simple rational number.

A promising solution to this problem
is to use a dual of lossless coding
where the encoding and the decoding are inverted.
Since a lossless code converts a biased sequence
into an almost uniform compressed sequence,
it is natural to expect that
a decoder of a lossless code can be used to generate
a biased codeword.
This framework is first considered in the literature
of LDPC codes \cite{miyake_ieice}\cite{miyake_channel_general}
and a similar idea is also proposed in polar codes \cite{polar_honda_trans}.
In these schemes fixed-length lossless (LDPC or polar) codes are used
to generate a biased codeword.
A similar scheme based on a fixed-length random number generator
is also found in \cite{channel_random_muramatsu}.

\subsection{Arithmetic Coding as a Biased-codeword Generator}\label{subs_arith}
Whereas such a fixed-length lossless code (or a fixed-length random number generator)
is convenient for theoretical analyses,
it is well known that
variable-length lossless code such as an arithmetic code
practically achieves the almost optimal performance.
Thus, it is natural to replace such fixed-length lossless codes
with variable-length ones.
In fact, coding schemes based on LDPC codes \cite{honda_lossy_trans}\cite{muramatsu_vf}
and polar codes \cite{wang_lossy}
where fixed-length lossless coding
is replaced with arithmetic coding have been proposed
in the context of lossy source coding, which can be regarded
as a dual of channel coding.

Nevertheless, a naive use of an arithmetic decoder
cannot be used as a generator of a biased codeword in channel coding.
To see this,
let us consider the code tree
in Fig.~\ref{tree_arithmetic}
of Shannon-Fano-Elias code (SFE code, see, e.g.~\cite[Sect.~5.9]{cover} for detail)
for
$(X_1,\dots,X_4) \in \{a,b\}^4$
i.i.d.~from
$(P_X(a),P_X(b))=(1/3,2/3)$.
\begin{figure}[t]%
  \begin{center}
   \includegraphics[bb=40 280 525 800,angle=270,clip, width=49mm]{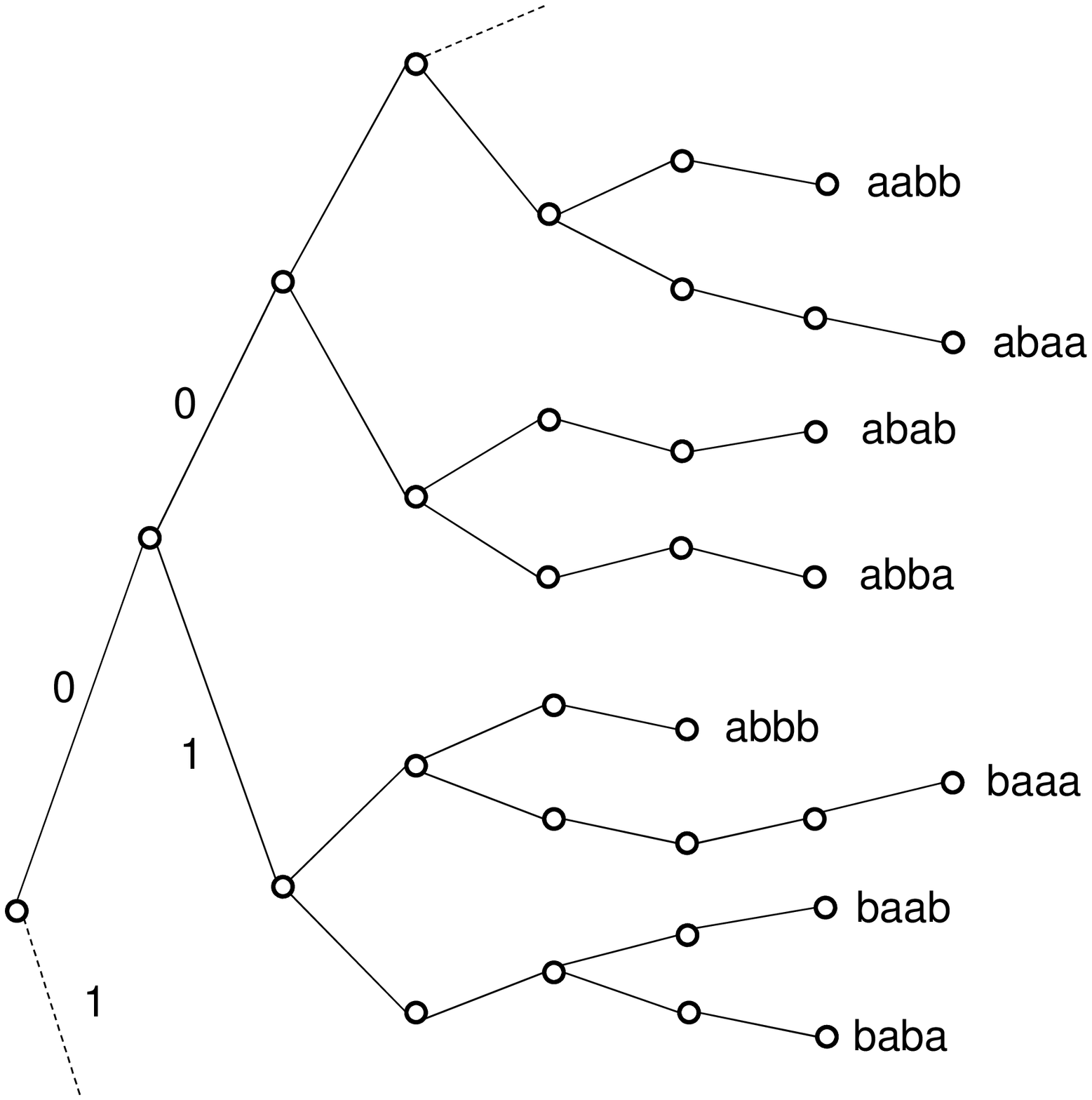}
  \end{center}
  \caption{A code tree of Shannon-Fano-Elias code.}
  \label{tree_arithmetic}
\end{figure}%
\begin{figure}[t]%
  \begin{center}
   \includegraphics[bb=40 280 530 675,angle=270,clip, width=37mm]{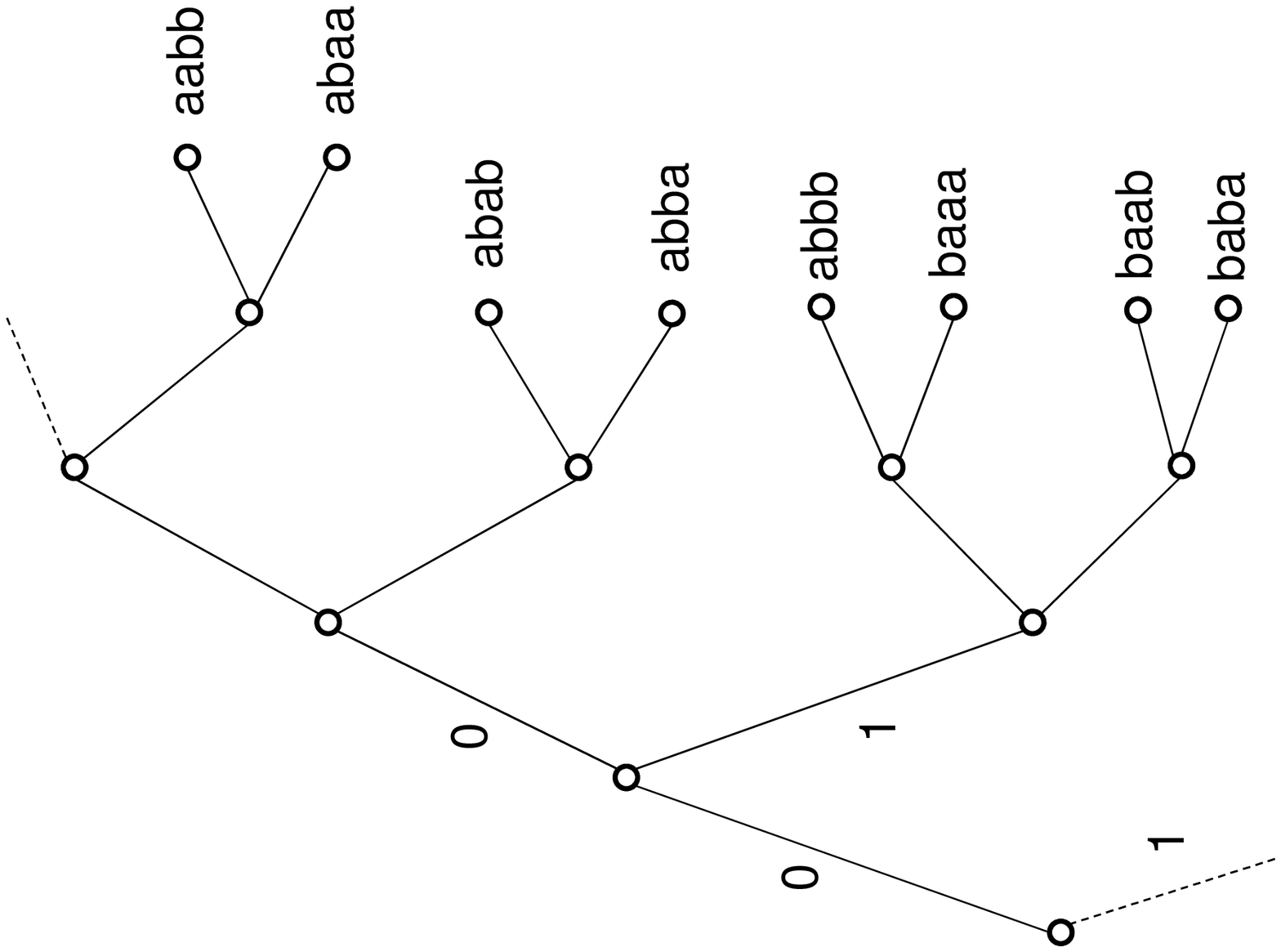}
  \end{center}
  \caption{A code tree of Shannon-Fano-Elias code without redundant edges.}
  \label{tree_arithmetic2}
\end{figure}%
As we can see from the figure,
the code tree of an arithmetic code is generally not complete
and cannot be used as a generator of
a sequence over $\{a,b\}$ from a binary sequence.
For example, the decoder of this code gets ``confused''
if it receives sequence 0111 since such a sequence
never appears as (a prefix of) a codeword of this code.
It is theoretically possible to consider a modified SFE code
where the redundant edges are removed as in Fig.~\ref{tree_arithmetic2}.
However
it is very difficult 
to realize such a code
with a linear complexity 
because
the code tree is practically not stored in the memory
and it is necessary to compute many other codewords
to find redundant edges.

Furthermore,
even if a modified SFE code without redundant edges is realized,
such a code is still not appropriate as a generator of a biased sequence.
For example in the code tree of Fig.~\ref{tree_arithmetic2}, sequence 0101 is converted into $baaa$ and therefore
$baaa$ appears with probability $1/16$
which is roughly 2.5 times larger than the target probability $P(baaa)=2/81$.
Such a problem generally occurs
because there sometimes exists a large gap between
target probabilities between adjacent sequences
under the lexicographic order as in
$P_{X^4}(abbb)$ and $P_{X^4}(baaa)$.
However,
sorting of sequences (as in Shannon-Fano code) is not practical
since the linear-time computability of
the cumulative distribution is lost.

\subsection{Homophonic Codes for Channel Coding}
Homophonic coding is another candidate for a generator of biased codewords.
This is a framework to
convert a sequence with distribution $P_U$ reversibly
into another sequence with the target distribution $P_X$.
In particular, a homophonic code is called {\it perfect}
if the generated sequence exactly follows the target distribution $P_X$.

Hoshi and Han \cite{homophonic_interval} proposed
a Fixed-to-Variable (FV) length perfect homophonic coding
scheme based on an interval algorithm similar to a random number generator \cite{interval}.
This code is applied to generation of biased codewords
in \cite{honda_phd}\cite{wang_memory}
but these FV channel codes suffer the following problem
for practical use.
When we use an FV homophonic code to generate a biased codeword with
block length $n$, an $m$-bit message
is sometimes converted into one block of codeword and
is converted into two blocks another time.
Thus, if a decoding error occurred in one block then the receiver can no more
know where the codeword is separated for each $m$-bit message
and the decoding error propagates to
all the subsequent sequences.
Based on this observation
it is desirable to use a Variable-to-Fixed (VF) length homophonic coding scheme
for a component of a channel code.

\subsection{VF Homophonic Coding by Gray Code}
Although it is difficult to realize a perfect VF homophonic code,
we can relax the problem
when we consider application
to channel coding.
Since homophonic coding is first considered in the context of
cryptography \cite{gunther},
the target distribution $P_X$ is usually uniform and
there is a special meaning to be perfect, that is,
the output is exactly random.
On the other hand in application to channel coding,
the output codeword does not have to perfectly follow the
target distribution
and it suffices to assure that
a bad codeword does not appear too frequently.

Keeping this difference of application in mind,
we propose a new VF homophonic code, dual SFEG code.
This code corresponds to a dual of
a modified SFE code
based on Gray code \cite{gray_patent}\cite{grey},
which we call SFEG code.
In SFEG code,
the cumulative distribution is defined
according to the order induced by Gray code
instead of the lexicographic order.
Under this order we can assure
linear-time computability of the cumulative distribution
and a small gap of probabilities between adjacent sequences.
Based on this property we prove that the dual SFEG code
is asymptotically perfect
and its coding rate is also asymptotically optimal.
We also prove as a by-product that
SFEG code for lossless compression
achieves a better coding rate 
than the original SFE code.


\section{Preliminaries}\label{sec_pre}
We use superscript $n$ to denote an $n$-bit sequence
such as $x^n=(x_1,x_2,\cdots,x_n)\in\cX^n$ for some alphabet $\cX$.
A subsequence is denoted by $x_i^j=(x_i,x_{i+1},\cdots,x_j)$ for $i\le j$.
Let $X^n=(X_1,X_2,\cdots,X_n) \in\cX^n$ be a discrete i.i.d.~sequence,
whose probability mass function is denoted by
$p(x)=\Pr[X_i=x],\,i=1,2,\cdots,n,$ and $p(x^n)=\Pr[X^n=x^n]=\prod_{i=1}^n p(x_i)$.
For notational simplicity we consider the case $\cX=\{0,1\}$
and assume $p(0),p(1)\in (0,1)$ and $p(0)\neq p(1)$.
The maximum relative gap
between probability masses is denoted by
$\gap=\max\{p(0)/p(1), p(1)/p(0)\}\in(1,\infty)$.
The cumulative distribution function is defined as
\begin{align}
F(x^n)=\sum_{a^n \prec x^n} p(x^n)\n
\end{align}
under some total order\footnote{In this paper we always write $x\prec y$
including the case $x=y$.}
$\prec$ over $\{0,1\}^n$.
We write $x^n+1$ for the next sequence to $x^n$, that is,
the smallest sequence $y^n$ such that
$x^n\precneqq y^n$. Sequence $x^n-1$ is defined
in the same way.

\subsection{Shannon-Fano-Elias Code}
Shannon-Fano-Elias (SFE) code is a lossless code
which encodes input $x_{(i)}^n,\,i=1,2,\cdots,$ into
$\phi_{\mathrm{SFE}}(x_{(i)}^n)=\break
\lfloor F(\lprec{x_{(i)}^n})+p(x_{(i)}^n)/2\rfloor_{\lceil -\log p(x_{(i)}^n)\rceil+1}$,
where $\lfloor r \rfloor_{l}$ for $l\in \mathbb{N}$
is the first $l$ bits of the binary expansion of $r\in [0,1)$.
When we define the cumulative distribution function $F$ for the lexicographic order
this code can be encoded and decoded with a complexity linear in $n$.
The expected code length satisfies
\begin{align}
\E[|\phi_{\mathrm{SFE}}(X^n)|]
&=
\sum_{x^n}p(x^n)\left(
\lceil \log p(x^n)\rceil+1
\right)
\nn&
<
n H(X)+2\com\label{upper_sfeg}
\end{align}
where $|u|$ for $u\in\{0,1\}^*$ denotes
the length of the sequence.

\subsection{Gray Code}

Gray code $\gray(\cdot)$ is
a one-to-one map over $n$-bit sequences.
This is the
XOR operation of
the input sequence and its one-bit shift to the right.
For example,
$g(0110)=0110\oplus 0011=0101$ and
$g(1101)=1101\oplus 0110=1011$,
where $\oplus$ is the bit-wise addition over
$\mathrm{GF}(2)$.
Table \ref{tab_gray} shows the output of Gray code for 3-bit sequences.
\begin{table}[b]%
\caption{Gray Code.%
}%
\begin{center}%
\label{tab_gray}%
\begin{tabular}{c|cccccccc}%
$x^n$&000&001&010&011&100&101&110&111\\
\hline
$\gray(x^n)$&000&001&011&010&110&111&101&100\\
\end{tabular}%
\end{center}
\end{table}%
The most important property of Gray code is that
if $x^n$ and $y^n$ are adjacent in the lexicographic order then
$g(x^n)$ and $g(y^n)$ differ only in one bit
as seen from the table.

We define {\it Gray order} $\precg$ as the total order
induced by Gray code, that is,
$x^n \precg y^n$ if and only if
$\gray^{-1}(x^n)\precl \gray^{-1}(y^n)$ where
$\precl$ is the lexicographic order.
For example, we have
$000\precg001\precg011\precg010\precg110\precg111\precg101\precg100$ from Table \ref{tab_gray}.
Gray order is represented in
a recursive way
\begin{align}
x^n\leg y^n
&\Leftrightarrow
\{x_1 \precneqq_{\mathrm{L}} y_1\}\cup
\{x_1=y_1=0,\,x_2^n \leg y_2^n\}\nn
&\phantom{wwwwwwwwc}\cup
\{x_1=y_1=1,\,y_2^n \leg x_2^n\}\per\n
\end{align}
From this expression
the cumulative distribution of $X^n$ under Gray order
can be computed in a linear time in $n$ as follows.
\begin{align}
\Pr[X^n\leg x^n]
&=
\begin{cases}
p(0)\Pr[X_2^n\leg x_2^n],&x_1=0,\\
p(0)+p(1)\Pr[X_2^n\geg x_2^n],&x_1=1,\\
\end{cases}\nn
\Pr[X^n\geg x^n]
&=
\begin{cases}
p(1)+p(0)\Pr[X_2^n\geg x_2^n],&x_1=0,\\
p(1)\Pr[X_2^n\leg x_2^n],&x_1=1.\\
\end{cases}\n
\end{align}

In the following we always assume that $x^n\in\{0,1\}^n$
is aligned by Gray order and
write $F(x^n)=\Pr[X^n\leg x^n]$ for the
cumulative distribution function under this order.
Similarly to the computation of $F(\cdot)$ we can show that
its inverse
$F^{-1}(r)=\min\{x^n: F(x^n)>r\}$
is also computed in a linear time.
From the property of Gray code we always have
\begin{align}
1/\gap\le p(x^n-1)/p(x^n)\le \gap\per\label{prop_gray}
\end{align}

\subsection{Homophonic Coding}
Let $\cU$ and $\cX$ be the input and output alphabets of sequences.
Let $\cI\subset \cU^*$ be a set such that
for any sequence $u^{\infty}\in \cU^{\infty}$
there exists $m>0$ such that $u_1^m\in \cI$.
A homophonic code $\phi$ is a (possibly random) map from
$\cI\subset \cU^*$ onto $\cO\subset \cX^*$.
%
A homophonic code is called {\it perfect} with respect to
the pair of distributions $(P_U,\,P_X)$
if
$\phi(U_{(1)}^{m_1})\phi(U_{(2)}^{m_2})\cdots$ is i.i.d.~from $P_X$ for 
the input $U_{(1)}^{m_1}U_{(2)}^{m_2}\cdots$ i.i.d.~from $P_U$.
We define that a homophonic code
is {\it weakly $\delta$-perfect}
in the sense of max-divergence
if
\begin{align}
\limsup_{k\to\infty}
\frac1k
\sup_{x^k: P_{\tilde{X}^k}(x^k)>0}
\log\frac{P_{\tilde{X}^k}(x^k)}{P_{X^k}(x^k)}\le 
\delta \com\label{def_weak}
\end{align}
where
$\tilde{X}^k$ is
the first $k$ symbols of the sequence
$\phi_n(U_{(1)}^{m_1})\phi_n(U_{(2)}^{m_2})\cdots$.
We call this notion ``weakly'' perfect
since
a perfect homophonic code is a code
satisfying the condition such that
$\delta=0$ and
$\limsup_{k\to\infty}$ is replaced with $\sup_{k\in\mathbb{N}}$
in \eqref{def_weak}.

A weakly perfect homophonic code
can be used as a component of a capacity-achieving
channel code in the following way.
Assume that there exists
a VF weakly $\delta$-perfect
homophonic code with output length $n$
and a channel code with block length $n'$ such that
the decoding error probability is $\epsilon_{n'}$ satisfying $\lim_{n'\to\infty}\epsilon_{n'}=0$
under some ideal codeword distribution.
Since the decoding error probability of the $n$-block sequence
of the channel codewords is at most $n\epsilon_{n'}$
under the ideal distribution, the decoding error probability
of the sequence generated by the VF homophonic code
is roughly bounded by $n\epsilon_{n'}2^{n\delta}$.
Thus the decoding error probability can be arbitrarily small
when $n'$ is sufficiently large with respective to $n$.
Based on this argument we can easily construct a VF channel
code achieving the capacity by, e.g., replacing the FV homophonic
code used as a component for the capacity-achieving
channel code in \cite{wang_memory} with such a weakly $\delta$-perfect
VF homophonic code.
In this paper we construct a weakly $\delta_n$-perfect
VF homophonic code
such that $\delta_n=O(1/n)$.

\section{Shannon-Fano-Elias-Gray Code}\label{sec_sfeg}
In this section
we propose Shannon-Fano-Elias-Gray (SFEG) code
as a simple modification of SFE code.
This encodes $x_{(i)}^n$ 
into 
$\phi_{\mathrm{SFEG}}(x_{(i)}^n)
=\lfloor F(\gprec{x_{(i)}^n})+ p(x_{(i)}^n)/2\rfloor_{\ls(x_{(i)}^n)}$,
where
$\ls(x^n)=\lceil -\log \als p(x^n)\rceil+1$ and $\als=(1+\gap)/\gap\in (1,2)$.
There are only two differences from the SFE encoder:
the cumulative distribution $F$ is defined by Gray order
and there is a factor $\als$ in the
code length $\ls(x^n)$.
The decoding of this code is given in Algorithm \ref{dec_sfeg}.
Here by abuse of notation
we sometimes identify
$\lfloor r\rfloor_l$, the first $l$ bits of the binary expansion of $r$,
with the real number $2^{-l}\lfloor 2^l r \rfloor$.

\begin{algorithm}[t]
\caption{Decoding of SFEG Code}\label{dec_sfeg}
\begin{algorithmic}[1]
\sonomama{Input:} Received sequence $u^{\infty}\in \{0,1\}^{\infty}$.
\STATE $i:=1,\,j:=1$.
\LOOP
 \STATE $r:=0.u_{j}u_{j+1}\cdots$ and $\hx^n:=F^{-1}(r)$.
 \STATE $\hl:=\ls(\hx^n)$ and $\hF=\lfloor \FG(\gprec{\hx^n})+p(\hx^n)/2\rfloor_{\hl}$.
 \IF{$r\ge \hF+2^{-\hl}$}
  \STATE $\tx^n:=\hx^n+1$.
 \ELSIF{$r< \hF$}
  \STATE $\tx^n:=\hx^n-1$.
 \ELSE
  \STATE $\tx^n:=\hx^n$.
 \ENDIF
 \STATE Output $\tx_{(i)}^n:=\tx^n$.
 \STATE $i:=i+1,\,j:=j+\ls(\tx^n)$.
\ENDLOOP
\end{algorithmic}
\end{algorithm}%

\begin{theorem}\label{thm_sfeg}
SFEG code is uniquely decodable.
Furthermore,
the average code length satisfies
\begin{align}
\E[|\phi_{\mathrm{SFEG}}(X^n)|]&<
nH(X)+2-\log
((1+\gap)/\gap)\per
\label{l_sfeg}
\end{align}
\end{theorem}
From this theorem we see that
the upper bound on the average code length of SFEG code improves that of
SFE code in \eqref{upper_sfeg} by $\log (1+\gap)/\gap\in(0,1)$.

We prove this theorem by the following lemma.
\begin{lemma}\label{lem_keta}
If $x+2^{-l'}\ge x'+2^{-\min\{l,l'\}}$ then
$\lfloor x\rfloor_l \ge \lfloor x'\rfloor_{l'}$.
\end{lemma}
\begin{proof1}
This lemma is straightforward from
\begin{align}
\lefteqn{
\lfloor x\rfloor_l \ge
\lfloor x'\rfloor_{l'}
}\nn
&\Leftarrow
\{l\ge l',\,x\ge x'\}\cup
\{l< l',\,x\ge x'+2^{-l}-2^{-l'}\}
\nn
&\Leftrightarrow
\{l\ge l',\,x+2^{-l'}\ge x'+2^{-l'}\}\nn
&\qquad\cup
\{l< l',\,x+2^{-l'}\ge x'+2^{-l}\}\nn
&\Leftrightarrow
x+2^{-l'}\ge x'+2^{-\min\{l,l'\}}\per
\dqed
\n
\end{align}
\end{proof1}

\begin{proof}[Proof of Theorem \ref{thm_sfeg}]
Eq.~\eqref{l_sfeg} holds since
\begin{align}
\E[|\phi_{\mathrm{SFEG}}(X^n)|]
&=
\sum_{x^n}p(x^n)
(\lceil -\log \als p(x^n)\rceil+1)
\nn
&<
\sum_{x^n}p(x^n)
(-\log \als p(x^n)+2)\nn
&=
nH(X)+2-\log
((1+\gap)/\gap)\per\n
\end{align}

We prove the unique decodability
by showing that
$\tx_{(i)}^n=x_{(i)}^n$ holds in Step \ref{xin} of Algorithm \ref{dec_sfeg}.
Let $x^n=x_{(i)}^n$ and $G=\phi_{\sfeg}(x^n)=\lfloor F(\gprec{x^n})+p(x^n)/2\rfloor_{\ls(x^n)}$.
Then
$r\in [G,G+2^{-\ls(x^n)})$ holds
from the encoding algorithm.
From the decoding algorithm, 
if $r \in [F(x^n-1),F(x^n))$ then $(\hx^n,\,\tx^n)$
given in Algorithm \ref{dec_sfeg} satisfies
$\hx^n=\tx^n=x^n$ and
we consider the other case in the following.

First we consider the case
$r \in [G, F(x^n-1))$.
Since 
$\hx^n\precneqq x^n$ in this case,
$\tx^n=x^n$ is equivalent to
$\{\hx^n=x^n-1,\,r\ge \hF+2^{-\hl}\}$,
where $\hF$ is given in Algorithm \ref{dec_sfeg}.
The former equality $\hx^n=x^n-1$ holds
since
\begin{align}
r\ge G
&\ge
F(x^n-1)+p(x^n)/2-2^{\log \als p(x^n)-1}\nn
&=
F(x^n-2)+p(x^n-1)+(1-\als)p(x^n)/2\nn
&\ge
F(x^n-2)+p(x^n)/\gap+(1-\als)p(x^n)/2\nn
&=
F(x^n-2)+p(x^n)/(2\rho)\nn
&\ge F(x^n-2)\per\n
\end{align}
We obtain the latter inequality $r\ge \hF+2^{-\hl}$
by letting $\bF(x^n)=F(x^n-1)+p(x^n)/2$
and using Lemma \ref{lem_floor}
since
\begin{align}
\lefteqn{
r\ge \hF+2^{-\hl}
}\nn
&\Leftarrow
\lfloor \bF(x^n)\rfloor_{\ls(x^n)}
\ge
\lfloor \bF(x^n-1)\rfloor_{\ls(x^n-1)}
+
2^{-\ls(x^n-1)}\nn
&\Leftarrow
\bF(x^n)-\bF(x^n-1)
\ge
2^{-\min\{\ls(x^n-1),\ls(x^n)\}}\nn
&\Leftrightarrow
(p(x^n-1)+p(x^n))/2
\ge
2^{-1-\lceil-\log \als\max\{p(x^n-1),p(x^n)\}\rceil}\nn
&\Leftarrow
(1+1/\gap)\max\{p(x^n-1),p(x^n)\}/2\nn
&\phantom{wwwwwwwwwwwwwwii}\ge
\als\max\{p(x^n-1),p(x^n)\}/2\nn
&\Leftrightarrow
1+1/\gap
\ge
\als
\per\n
\end{align}

Finally we consider the remaining case $r \in [F(x^n),\allowbreak G+2^{-\ls(x^n)})$.
In the same way as the former case we can show $\{\hx^n=x^n+1,\,r< \hF\}$,
which implies $\tx^n=x^n$.
\end{proof}

\section{Dual SFEG Code}\label{sec_dsfeg}
In this section we construct a VF homophonic code
based on SFEG code, which we call the dual SFEG code.
The main difference from SFEG code
is that we use the transformed cumulative distribution function
$F_I(x^n)=a+(b-a)F(x^n)$ for an interval $I=[a,b)\subset [0,1)$.
As we show later, the real number $r$ corresponding to the message sequence at each iteration
is uniformly distributed over an interval $I$ that is generally different from $[0,1)$.
By using the transformed function $F_{I}$ the output distribution becomes close to $P_{X^n}$.
Let
$\res{r}{l}=2^{l}(r-\lfloor r\rfloor_l)\in[0,1)$
be the real number corresponding to the $(l+1,l+2,\cdots)$-th bits
of $r\in[0,1)$.
The encoding and decoding of the dual SFEG code,
which are similar to the decoding and encoding
of SFEG code,
are given in
Algorithms \ref{enc_dsfeg} and \ref{dec_dsfeg},
respectively,
where
\begin{align}
\oFI(x^n)&=\lfloor F_I(x^n-1)+2^{-\lh(x^n)}\rfloor_{\lh(x^n)}\n
\end{align}
for $\lh(x^n)=\lfloor -\log \alh(b-a) p(x^n) \rfloor$ and $I=[a,b)$.
\begin{algorithm}[t]%
\caption{Encoding of Dual SFEG Code}\label{enc_dsfeg}
\begin{algorithmic}[1]
\sonomama{Input:} Message $u^{\infty}\in \{0,1\}^{\infty}$.
\STATE $i:=1,\,j:=1,\,I:=[0,1)$.
\LOOP 
 \STATE $r:=0.u_{j}u_{j+1}\cdots,\,\hx^n:=F_I^{-1}(r)$. \label{step_r}
 \IF{$r\ge \oFI(\hx^n)$}
  \STATE $\tx^n:=\hx^n+1$.
 \ELSE
  \STATE $\tx^n:=\hx^n$.
 \ENDIF
 \STATE Output $\tx_{(i)}^n:=\tx^n$.\label{xin}
 \STATE $i:=i+1,\,j:=j+\lh(\tx^n)$.\label{step}
 \STATE $I:=\big[\res{\min\{\oF_{I}(\tx^n-1),F_{I}(\tx^n-1)\}}{\lh(\tx^n)},\qquad\quad\break
 \phantom{wwwwwwwwwa}\res{\min\{\oF_{I}(\tx^n),F_{I}(\tx^n)\}}{\lh(\tx^n)}\big)$.
\ENDLOOP
\end{algorithmic}
\end{algorithm}%
\begin{algorithm}[t]%
\caption{Decoding of Dual SFEG Code}\label{dec_dsfeg}
\begin{algorithmic}[1]
\sonomama{Input:} Received sequences $\tx_{(1)}^n,\,\tx_{(2)}^n,\cdots\in\{0,1\}^n$.
\STATE $i:=1,\,j:=1,\,I:=[0,1)$.
\LOOP 
 \STATE Output
$\hat{u}_{j}^{j+\lh(\tx_{(i)}^n)-1}:=\lfloor F_I(\tx^n-1)\rfloor_{l(\tx_{(i)}^n)}$.
 \STATE $I:=\big[\res{\min\{\oF_{I}(\tx^n-1),F_{I}(\tx^n-1)\}}{\lh(\tx^n)},\qquad\quad\break
 \phantom{wwwwwwwwww}\res{\min\{\oF_{I}(\tx^n),F_{I}(\tx^n)\}}{\lh(\tx^n)}\big)$.
 \STATE $i:=i+1,\,j:=j+\lh(\tx_{(i)}^n)$.
\ENDLOOP
\end{algorithmic}
\end{algorithm}%
\begin{theorem}\label{thm_dual}
Dual SFEG code for uniform input $U^{\infty}$
with outputs $\tilde{X}_{(1)}^n,\tilde{X}_{(2)}^n\cdots$
is (i) uniquely decodable,
(ii) weakly $(\frac1n\log2\alh)$-perfect
and
(iii) the average input length satisfies
\begin{align}
\E[\lh(\tilde{X}^n)]
>nH(X)-1-2\log \gap\per\n
\end{align}
\end{theorem}
In the dual SFEG code
the property \eqref{prop_gray} of Gray code is essentially used to prove
the code is weak $\delta$-perfect for some $\delta\in(0,\infty)$.
This fact contrasts with the relation between SFE code and SFEG code,
where Gray code only contributes to improve the code length by
$(1+\gap)/\gap$.

Now define
$I(x^n)=[\min\{\oF_{I}(x^n-1),F_{I}(x^n-1)\},\allowbreak\min\{\oF_{I}(x^n),F_{I}(x^n)\})$.
We show Lemmas \ref{lem_floor}--\ref{lem_uniform} in the following to prove the theorem.
\begin{lemma}\label{lem_floor}
For any $x^n\in\{0,1\}^n$ it holds that
$\oFI(x^n)\ge
\max\{
F_{I}(x^n-1),\allowbreak
\lfloor F_{I}(x^n)\rfloor_{\lh(x^n+1)}
\}$
and, consequently,
\begin{align}
I(x^n)\subset [\lfloor F_I(x^n-1)\rfloor_{\lh(x^n)},\,\oF_{I}(x^n))\per\label{lem_include}
\end{align}
\end{lemma}
\begin{proof1}
$\oFI(x^n)\ge
F_{I}(x^n-1)$
is straightforward from the definition of $\oFI(x^n)$
and it suffices to show
$\oFI(x^n)\ge\lfloor F_{I}(x^n)\rfloor_{\lh(x^n+1)}$,
which is equivalent to
\begin{align}
\lfloor G+2^{-\lh(x^n)} \rfloor_{\lh(x^n)}
\ge
\lfloor G+(b-a)p(x^n)\rfloor_{\lh(x^n+1)}\com\label{lem_equi}
\end{align}
where $G=F_I(x^n-1)$.
This holds from Lemma \ref{lem_floor} since
\begin{align}
\eqref{lem_equi}&\Leftarrow
G
+2^{-\lh(x^n)}+2^{-\lh(x^n+1)}\nn
&\qquad\qquad\ge G+
(b-a)p(x^n)+2^{-\min\{\lh(x^n),\lh(x^n+1)\}}
\nn
&\Leftrightarrow
2^{-\max\{\lh(x^n),\lh(x^n+1)\}}\ge (b-a)p(x^n)
\nn
&\Leftrightarrow
\alh \min\{p(x^n),p(x^n+1)\}\ge p(\hx^n)
\nn
&\Leftarrow
\alh p(x^n)/\gap\ge p(x^n)\per\dqed\n
\end{align}
\end{proof1}

\begin{lemma}\label{lem_interval}
At each loop of Algorithm \ref{enc_dsfeg},
$\tx^n=x^n$ if and only if
$r\in I(x^n)$.
\end{lemma}
\begin{proof}
Since $\hx^n =x^n$ if and only if
$r \in [F_I(x^n-1),F_I(x^n))$,
we have $\tx^n=x^n$ if and only if
\begin{align}
&\{r\in [F_I(x^n-2), F_I(x^n-1)),\allowbreak r\ge \oFI(x^n-1)\} \mbox{ or } \nn
&\{r\in [F_I(x^n-1), F_I(x^n)),\,r< \oFI(x^n)\}\per\label{cond12}
\end{align}
Since $\oFI(x^n-1)\ge F_I(x^n-2)$ holds from Lemma \ref{lem_keta},
\eqref{cond12} is equivalent to
\begin{align}
&\oFI(x^n-1)\le r < F_I(x^n-1)\mbox{ or } \nn
&F_I(x^n-1)\le r < \min\{F_I(x^n),\,\oFI(x^n)\}\per\label{cond22}
\end{align}

We can easily show that $\eqref{cond22}$ is equivalent to
$r\in I(x^n)$ by considering
cases $\oFI(x^n-1)\le F_I(x^n-1)$
and $\oFI(x^n-1)> F_I(x^n-1)$ separately.
\end{proof}

\begin{lemma}\label{lem_uniform}
At Step \ref{step_r} in each loop of Algorithm \ref{enc_dsfeg},
$r$ is uniformly distributed over $I$.
\end{lemma}
\begin{proof}
In the first loop $r$ is uniformly distributed over $I=[0,1)$.
Assume that $r$ is uniformly distributed over $I$ at some loop.
Then, given $\tx^n$ is sent,
$r$ is uniformly distributed over
$I(\tx^n)$
from Lemma \ref{lem_interval}.
Here we have from Lemma \ref{lem_floor} that $I(\tx^n)\subset [\lfloor F_I(\tx^n-1)\rfloor_{\lh(\tx^n)},\,\oF_{I}(\tx^n))$,
which implies that the first $\lh(\tx^n)$ bits of $r=0.u_{j}u_{j+1}\cdots\in I(\tx^n)$ are unique.
As a result, $0.u_{j+\lh(\tx^n)}u_{j+\lh(\tx^n)+1}\cdots$ is uniformly distributed over
$\big[\res{\min\{\oF_{I}(\tx^n-1),F_{I}(\tx^n-1)\}}{\lh(\tx^n)},\,
\res{\min\{\oF_{I}(\tx^n),F_{I}(\tx^n)\}}{\lh(\tx^n)}\big)$, which means that
$r$ is also uniformly distributed over $I$ in the next loop.
\end{proof}

\begin{proof2}{of Theorem \ref{thm_dual}}
(i)
By the argument in the proof of Lemma \ref{lem_uniform},
$(u_{j},u_{j+1},\,\cdots,\,u_{j+\lh(\tx_{(i)}^n)-1})$ is given as
$\lfloor F_I(\tx_{(i)}^n-1)\rfloor_{\lh(\tx_{(i)}^n)}$, which implies the unique decodability.


(ii)
Let $\tp_{(i)}(x^n)$ be the probability of the event $\tx_{(i)}^n=x^n$.
Then it holds from \eqref{lem_include} and Lemmas \ref{lem_interval} and \ref{lem_uniform} that
\begin{align}
\tp_{(i)}(x^n)
\le \frac{2^{-\lh(x^n)}}{b-a}
= \frac{2^{-\lfloor -\log \alh (b-a)p(x^n)\rfloor}}{b-a}
&< 2\alh p(x^n)\per\n
\end{align}
Therefore, letting $\tilde{X}^k$ be the first $k$ bits of the output sequence
$\tx_{(1)}^n,\tx_{(2)}^n,\cdots$ for the uniform input $U^{\infty}$,
we have
\begin{align}
\limsup_{k\to\infty}\frac1k \sup_{x^k}\log\frac{P_{\tilde{X}^k}(\tilde{X}^k)}{P_{X^k}(\tilde{X}^k)}
&\le
\limsup_{k\to\infty}\frac{1}{k} \log (2\alh)^{\lceil k/n\rceil}
\nn
&=
\frac{1}{n}
\log 2\alh\com\n
\end{align}
that is, the dual SFEG code is weakly $(\frac1n\log2\alh)$-perfect.

(iii) Let $\tilde{F}_{(i)}(x^n)$ be the cumulative distribution function
induced by $\tp_{(i)}(x^n)$.
Then
$F(x^n-1)<\tilde{F}_{(i)}(x^n)\le F(x^n)$ holds
and therefore
\begin{align}
\E[\lh(\tilde{X}_{(i)}^n)]
&=
\int_0^1 \lh(\tF_{(i)}^{-1}(r))\rd r\nn
&\ge
\int_0^1 \min\{\lh(F^{-1}(r)-1),\lh(F^{-1}(r))\}\rd r\nn
&=
\int_0^1
\lfloor -\log \alh \max\{p(F^{-1}(r)-1),p(F^{-1}(r))\}\rfloor \rd r\nn
&\ge
\int_0^1
\lfloor -\log \gap^2 p(F^{-1}(r))\rfloor \rd r\nn
&=
\sum_{x^n}p(x^n)\lfloor -\log \gap^2 p(x^n)\rfloor\nn
&>
nH(X)-1-2\log \gap\per\n
\dqed
\n
\end{align}
\end{proof2}

\bibliographystyle{IEEEtran}

\begin{thebibliography}{10}
\providecommand{\url}[1]{#1}
\csname url@samestyle\endcsname
\providecommand{\newblock}{\relax}
\providecommand{\bibinfo}[2]{#2}
\providecommand{\BIBentrySTDinterwordspacing}{\spaceskip=0pt\relax}
\providecommand{\BIBentryALTinterwordstretchfactor}{4}
\providecommand{\BIBentryALTinterwordspacing}{\spaceskip=\fontdimen2\font plus
\BIBentryALTinterwordstretchfactor\fontdimen3\font minus
  \fontdimen4\font\relax}
\providecommand{\BIBforeignlanguage}[2]{{%
\expandafter\ifx\csname l@#1\endcsname\relax
\typeout{** WARNING: IEEEtran.bst: No hyphenation pattern has been}%
\typeout{** loaded for the language `#1'. Using the pattern for}%
\typeout{** the default language instead.}%
\else
\language=\csname l@#1\endcsname
\fi
#2}}
\providecommand{\BIBdecl}{\relax}
\BIBdecl

\bibitem{gallager_map}
R.~G. Gallager, \emph{Information Theory and Reliable Communication}.\hskip 1em
  plus 0.5em minus 0.4em\relax New York: Wiley, 1968.

\bibitem{miyake_ieice}
S.~Miyake and J.~Muramatsu, ``A construction of lossy source code using {LDPC}
  matrices,'' \emph{IEICE Trans. Fundam.}, vol. 91-A, pp. 1488--1501, 2008.

\bibitem{miyake_channel_general}
------, ``A construction of channel code, joint source-channel code, and
  universal code for arbitrary stationary memoryless channels using sparse
  matrices,'' \emph{IEICE Trans. Fundam.}, vol. 92-A, no.~9, pp. 2333--2344,
  2009.

\bibitem{polar_honda_trans}
J.~Honda and H.~Yamamoto, ``Polar coding without alphabet extension for
  asymmetric models,'' \emph{IEEE Trans. Inform. Theory}, vol.~59, no.~12, pp.
  7829--7838, 2013.

\bibitem{channel_random_muramatsu}
J.~Muramatsu, ``Channel coding and lossy source coding using a generator of
  constrained random numbers,'' \emph{IEEE Trans.~Inform.~Theory}, vol.~60,
  no.~5, pp. 2667--2686, 2014.

\bibitem{honda_lossy_trans}
J.~Honda and H.~Yamamoto, ``Variable length lossy coding using an {LDPC}
  code,'' \emph{IEEE Trans. Inform. Theory}, vol.~60, no.~1, pp. 762--775,
  2014.

\bibitem{muramatsu_vf}
J.~Muramatsu, ``Variable-length lossy source code using a
  constrained-random-number generator,'' in \emph{IEEE ITW2014}, Nov 2014, pp.
  197--201.

\bibitem{wang_lossy}
R.~Wang, J.~Honda, H.~Yamamoto, and R.~Liu, ``{FV} polar coding for lossy
  compression with and improved exponent,'' in \emph{IEEE ISIT2015}, 2015, pp.
  1517--1521.

\bibitem{cover}
T.~M. Cover and J.~A. Thomas, \emph{{Elements of Information Theory}},
  2nd~ed.\hskip 1em plus 0.5em minus 0.4em\relax Wiley-Interscience, July 2006.

\bibitem{homophonic_interval}
M.~Hoshi and T.~S. Han, ``Interval algorithm for homophonic coding,''
  \emph{IEEE Trans. Inform. Theory}, vol.~47, no.~3, pp. 1021--1031, 2001.

\bibitem{interval}
T.~S. Han and M.~Hoshi, ``Interval algorithm for random number generation,''
  \emph{IEEE Trans. Inform. Theory}, vol.~43, no.~2, pp. 599--611, 2006.

\bibitem{honda_phd}
\BIBentryALTinterwordspacing
J.~Honda, ``Efficient polar and {LDPC} coding for asymmetric channels and
  sources,'' Ph.D. dissertation, The University of Tokyo, 2013. [Online].
  Available:
  \url{http://repository.dl.itc.u-tokyo.ac.jp/dspace/bitstream/2261/56414/1/K-04103.pdf}
\BIBentrySTDinterwordspacing

\bibitem{wang_memory}
R.~Wang, J.~Honda, H.~Yamamoto, R.~Liu, and Y.~Hou, ``Construction of polar
  codes for channels with memory,'' in \emph{IEEE ITW2015}, 2015, pp. 187--191.

\bibitem{gunther}
C.~G. G\"{u}nther, ``A universal algorithm for homophonic coding,'' in
  \emph{EUROCRYPT '88}.\hskip 1em plus 0.5em minus 0.4em\relax Springer, 1988,
  pp. 405--414.

\bibitem{gray_patent}
\BIBentryALTinterwordspacing
F.~Gray, ``Pulse code communication,'' Mar.~17 1953, {US Patent 2,632,058}.
  [Online]. Available: \url{http://www.google.com/patents/US2632058}
\BIBentrySTDinterwordspacing

\bibitem{grey}
\BIBentryALTinterwordspacing
P.~E. Black, ``Gray code,'' in \emph{Dictionary of Algorithms and Data
  Structures}, Jun. 2014. [Online]. Available:
  \url{http://www.nist.gov/dads/HTML/graycode.html}
\BIBentrySTDinterwordspacing

\end{thebibliography}


%

\end{document}